\documentclass[letterpaper, 10 pt, conference]{ieeeconf}
\IEEEoverridecommandlockouts  
\usepackage{setspace}
\usepackage{graphics} % for pdf, bitmapped graphics files
\usepackage{epsfig} % for postscript graphics files
\usepackage{times} % assumes new font selection scheme installed
\usepackage{amsmath} % assumes amsmath package installed
\usepackage{amssymb}  % assumes amsmath package installed
\usepackage{subfigure}
\usepackage{epstopdf}
\usepackage{cite}
\usepackage{verbatim}
\usepackage{changes}
\usepackage{xcolor}
\usepackage{bbm}
\usepackage{tikz}
\usepackage{caption}
\usepackage{algpseudocode}
\usepackage{mathtools}

\usepackage{amsthm}
\usepackage{pgfplots}
\usetikzlibrary{intersections}
\usetikzlibrary{patterns}
\usepackage[top=1in, bottom=1in, left=1in, right=1in]{geometry}
\usetikzlibrary{automata,positioning,shapes,arrows}

\newtheorem{theorem}{Theorem}
\newtheorem{lemma}{Lemma}
\usepackage[lined,ruled,commentsnumbered]{algorithm2e}

\newtheorem{definition}{Definition}

\newtheorem{example}{Example}
\title{Privacy Preserving Controller Synthesis via Belief Abstraction}
\author{Bo~Wu and Hai~Lin
	\thanks{The partial support of the National Science Foundation (Grant No. CNS-1446288, ECCS-1253488, IIS-1724070) and of the Army Research Laboratory (Grant No. W911NF- 17-1-0072) is gratefully acknowledged.}
	\thanks{ Bo Wu and Hai Lin are with the Department of Electrical Engineering, University of Notre Dame, Notre Dame,
		IN, 46556 USA. {\tt\small bwu3@nd.edu, hlin1@nd.edu}}}% <-this % stops a space
\date{}

\begin{document}

\maketitle
\begin{abstract}
Privacy is a crucial concern in many systems in addition to their given tasks. We consider a new notion of privacy based on beliefs of the system states, which  is closely related to opacity in discrete event systems. To guarantee the privacy requirement, we propose to abstract the belief space whose dynamics is shown to be mixed monotone where efficient abstraction algorithm exists. Based on the abstraction, we propose two different approaches to synthesize controllers of the system to preserve privacy with an illustrative example.
\end{abstract}
\section{Introduction}
Privacy is becoming one of the most critical concerns in many practical systems \cite{mcdaniel2009security,chan2003security,weber2010internet,hubaux2004security}. The vulnerabilities to information leaking pose significant challenges in systems that may have a huge social or economic impact if their privacy is compromised. Examples of such systems include automobiles, transportation systems, healthcare systems, robotic systems, power grid and so on. %Development and adaptation of privacy preserving planning framework is receiving increasing interest in both academia and industry due to its great potential to significantly improve the system reliability. 

In the recent years, a notion called ``opacity" is receiving an increasing interest in privacy analysis and enforcement. Opacity is a confidentiality property that characterizes a system's capability to hide its ``secret" information from being inferred by outside passive observers with possibly malicious intentions (termed as intruders in the sequel). The intruder is assumed to know the system's structure and has (partial) access to the system's outputs but cannot  observe the system states. The system is opaque if the intruder never decides that the secret happens with absolute certainty. 

%The opacity problem was first introduced in the computer science community \cite{mazare2004using} and quickly spread to Petri-net and discrete event system (DES) researchers. 
Various notions of opacity have been proposed in both deterministic and stochastic models. Interested readers are referred to \cite{jacob2016overview} for a comprehensive review. In  this paper, we are interested in the current-state opacity (CSO), where the secret information  is whether or not the current state of the system is a secret state. There are essentially two main directions in the opacity research --- verification and enforcement. Algorithms are designed to verify if the system is opaque from the intruders \cite{lin2011opacity}. And to enforcing the opacity, the proposed approaches  include  synthesizing the supervisor \cite{saboori2012opacity}, insertion functions \cite{wu2014synthesis,WuBo2017PSynthesis,WuBo2017Synthesis} or edit functions \cite{wuy2017synthesis} to control or manipulate the observed behavior .

The current definition of the current-state opacity relies on the absolute certainty that the current state belongs to the secret states. However, with a probabilistic model, in some cases the intruder may just be able to maintain a belief distribution over the system states. In other words, the intruder may only infer that the current state is a secret state with certain probability based on the observation history. As mentioned in \cite{saboori2014current}, such scenario may not be characterized as a CSO violation by its definition, but still may potentially pose security threat if the intruder deems the current state being a secret state with a high confidence. 

Thus, we are motivated to introduce, to the best of our knowledge, a new opacity notion where the system is considered  opaque if the intruder's confidence that the current state is a secret state never exceeds a given threshold. Similar privacy problems have been considered in the computer science community. Protecting users' anonymity on the world-wide-web by clustering the users in large groups that collectively issue requests for
the group members is studied in \cite{reiter1998crowds}, where the developed anonymity protocol hides the user identity which originates certain actions, such that the probability of  the sender being the originator based on the observed outputs satisfies certain property. Program synthesis to protect data privacy defined in intruder's belief is studied in \cite{kuvcera2017synthesis} where the intruder can interact with the program. The enforcement modifies the program by conflating the outputs if the privacy requirement is to be violated.

Typically, the intruder updates its belief by computing its posterior belief distribution based on its  a prior belief. Such update depends on the action executed by the system since it determines the transition probability. %Therefore, the belief update can be seen as a switched system that switches between the modes at each discrete time step where the continuous belief dynamics is governed by the action observed in each mode. 
The opacity requirement defines a convex region that the belief state should avoid. But to analyze whether the belief %represented by the switched system 
will always stay in the ``safe'' zone and satisfies privacy requirement could be a challenging task. In this paper, we propose to abstract the continuous belief space into a finite set of grids. By proving that the belief dynamics is mixed monotone, we could efficiently obtain the abstracted finite state system that serves as an over-approximation of the underlying continuous dynamic \cite{coogan2015efficient}, which has been successfully applied on the traffic network control with temporal logic specifications \cite{coogan2017formal}. The belief abstraction idea has also been proposed in \cite{bharadwaj2017synthesis}, but their belief space is the power set of the state space, which is discrete and finite.

With the abstracted finite belief transition system, we propose two different approaches to synthesize controllers to guarantee the privacy and optimize the given task specification, for example, in linear temporal logic (LTL) or probabilistic computation tree logic (PCTL) \cite{baier2008principles}. The first approach identifies the actions in each state that are guaranteed to preserve privacy and then synthesize the controller. The second approach is inspired by the edit function idea \cite{wuy2017synthesis} and directly manipulates the observations to the intruder, such that the intruder may never be confident that the system is currently in a secret state with the probability more than some threshold.

The rest of the paper is organized as follows. Section \ref{sec:pre} provides the necessary preliminaries to define and solve our problem. Section \ref{sec:opacity} introduces our opacity notion. Section \ref{sec:abstraction} deals with the efficient abstraction of the belief space based on the mixed monotone property. Section \ref{sec:synthesis} propose two approaches to obtain the controller that preserves the opacity and satisfies the task specification. Section \ref{sec:conclusion} concludes the paper.
%Opacity can guarantee the privacy of the cyber component, but it is computationally complex and does not consider the privacy and dynamics of the physical plants.

\section{Preliminaries}\label{sec:pre} 
%In this section, we introduce  Nondeterministic Finite Automaton and Markov Decision Process. 
\subsection{Nondeterministic Finite Automaton (NFA)}
NFA is a popular model to describe the non-probabilistic behavior of the system. 
\begin{definition}\cite{baier2008principles}
	An NFA is a tuple $\mathcal{T}=(Q,\Sigma,\delta,I)$ where 
	\begin{itemize}
	    \item $Q$ is a finite set of states; 
	    \item $\Sigma$ is a finite set of actions; 
	    \item $\delta:Q\times\Sigma\rightarrow 2^Q$ is the transition function; \item $I\subseteq Q$ is a set of initial states.
	\end{itemize}
	%$Q$ is a finite set of states; $\Sigma$ is a finite set of actions; $\delta:Q\times\Sigma\rightarrow 2^Q$ is the transition function; $I\subseteq Q$ is a set of initial states.
	    %\item $AP$ is a set of atomic propositions;
		%\item $L:Q\rightarrow2^{AP}$ is a labeling function that maps each $q\in Q$ to one or several elements of $AP$
\end{definition}
%We denote $Post(q,\sigma)=\{q'\in Q|(q,\sigma,q')\in \rightarrow\}$, as the set of states that are reachable from $q\in Q$ with $\sigma\in \Sigma$. $\Sigma(q)\subseteq\Sigma$ denotes all the actions that are defined on the state $q\in Q$. 
Note that we didn't define the accepting states, which is a subset of $Q$, since they are not of interest in this paper. The transition function $\delta$ can be extended to $Q\times\Sigma^*$ in a natural way. Given the initial set $I\subseteq Q$ of states, the language generated by $\mathcal{T}$ is defined by $\mathcal{L}(\mathcal{T})=\{\omega\in \Sigma^*|\exists q\in I,\delta(q,\omega)\text{ is defined}\}$. 
%Specifically, $\mathcal{T}$ is called an action-deterministic transition system if $|I|= 1$ and $|Post(q,\sigma)|\leq 1$ for all $q$ and $\sigma$. Intuitively speaking, it implies that there is a single initial state, for each state $q$ with any action $\sigma$, it will transit to no more than one state.
\subsection{Markov Decision Process}

\begin{definition}\label{def:MDP}\cite{puterman2014markov}
	An MDP  is a tuple $\mathcal{M}=(S,\pi_0,A,P)$ where
	\begin{itemize}
	    \item $S$ is a finite set of states;
	    \item $\pi_0:S\rightarrow [0,1],\sum_{s\in S}\pi_0(s)=1$,  is the initial state distribution; 
	    \item $A$ is a finite set of actions;
	    \item $P(s,a,s'):=Pr(s'|s,a)$. That is, the probability of transiting from $s$ to $s'$ with action $a$.
	\end{itemize}
	%$S$ is a finite set of states; $\pi_0:S\rightarrow [0,1],\sum_{s\in S}\pi_0(s)=1$,  is the initial state distribution; $A$ is a finite set of actions; $P(s,a,s'):=Pr(s'|s,a)$. That is, the probability of transiting from $s$ to $s'$ with action $a$.
\end{definition}

$A(s)$ denotes the set of available actions at the state $s$. In this paper, we assume that $A(s)=A,\forall s\in S$. %From the definition, it is not hard to see that% the action-deterministic transition system is a special case of MDP with $P(s,a,s')=1$ for any $s\in S$ and $a\in A$ that are defined.% Vice versa, 
If we ignore the transition probabilities, the MDP will become an NFA which we denote as $\mathcal{T}_\mathcal{M}=(S,A,\delta,I)$, where $s\in I$ if $\pi_0(s)>0$ and $s'\in\delta(s,a)$, if $P(s,a,s')>0$.

\section{Current State Opacity in Belief Space}\label{sec:opacity}
Given a system modeled as an MDP $\mathcal{M}=(S,\pi_0,A,P)$, we assume that there is an intruder that has the knowledge of $\mathcal{M}$ and is capable of observing all the actions  but not the actual states. Note that the state is fully observable for the policy of the system to make decisions. Such scenario could happen in web-based service or robotic applications where the internal states are hidden  but the service request or robot executions can be eavesdropped. In this case, the intruder may  maintain a belief $b_t:S\rightarrow [0,1],\sum_{s\in S}b_t(s)=1$ over $S$ at time $t$. At time $t+1$, when action $a\in A$ is observed, the belief update is as follows.
\begin{equation}
    b_{t+1}(s') = \sum_{s\in S}P(s,a,s')b_t(s)
\end{equation}
Equivalently in matrix form, we could have
\begin{equation}\label{eqn:belief dynamic}
    b_{t+1}=H_ab_t
\end{equation}
where $H_a$ is a $N\times N$ matrix with $H_a(i,j)=P(s_j,a,s_i)$, $b_0=\pi_0, N=|S|$. Therefore, the dynamics of the belief $b_t$ is governed by a switched linear system with $|A|$ modes. At any time $t$, it may choose to switch to some mode (action) $a\in A$. Suppose there are a subset of states $S_s\subset S$ representing the secret states that the system tries to hide from the intruder. $S_s$ is a strict subset of $S$, since if $S_s=S$ the problem will become trivial.  It is desirable that at any time, the intruder may never be sure that the system is in some secret state with probability over a threshold $\lambda\in[0,1]$. In other words,
\begin{equation}\label{eqn:opacity}
   \sum_{s\in S_s}b_t(s)\leq \lambda, \forall t
\end{equation}
Any belief state that violates (\ref{eqn:opacity}) is a bad state that should  be avoided. The switched linear system in (\ref{eqn:belief dynamic}) is analogous to the observers for a partially observed automaton \cite{cassandras2009introduction} whose states, instead of being a distribution over S, belong to $2^S$, the power set of $S$. The following motivating example will be used through out the paper to illustrate our framework.

\begin{example}\label{example}
Suppose the MDP $\mathcal{M}$ models the evolution of inventory levels of a company, which has three states, where $s_1$ and $s_2$ represents low and high inventory level and $s_3$ represents the medium inventory level. The company  would like to keep the current inventory level being too high or too low as secret, because the intruders, suppliers or competitors, may leverage such information to adjust the price of the goods for their own benefits. Therefore $s_1,s_2\in S_s$ and $x_1=b(s_1),x_2=b(s_2)$, $s_3$ is a non-secret state. $A = \{\sigma_1,\sigma_2\}$ represents two different purchase quantities. The initial condition is that $b(s_1)=0.3,b(s_2)=0.1$. The transition probabilities are as shown in the following matrices, because of random demand levels. 
\begin{equation}\label{eqn:transtion probability}
H_{\sigma_1}=\begin{bmatrix}
   0.2, &0 ,&0.1 \\
   0.4,&0.3,&0.2 \\
   0.4,&0.7,&0.7
\end{bmatrix},
H_{\sigma_2}=\begin{bmatrix}
   0.4, &0.65, &0.3 \\
   0.2, &0,    &0.2 \\
   0.4, &0.35,  &0.5
\end{bmatrix}    
\end{equation}
\end{example}

\section{Belief Abstraction}\label{sec:abstraction}
Checking whether the belief state will enter an undesired region by violating (\ref{eqn:opacity}) is a reachability problem of (\ref{eqn:belief dynamic}). %It could be very difficult to solve since the  belief space is continuous and there are infinite number of belief states. Therefore, 
In this paper, we explore the intrinsic structure of the system (\ref{eqn:belief dynamic}) by  showing that it is in fact mixed monotone where efficient abstraction method is available \cite{coogan2015efficient}. 
\begin{definition}\label{def:mixed monotone}
A system 
\begin{equation}\label{eqn:mixed monotone}
x = F(x)
\end{equation}
is mixed monotone, where $x\in X\subset \mathbb{R}^n$ and $F:X\rightarrow X$ is a continuous map, if there exists a decomposition function $f:X\times X\rightarrow X$ such that 1) $F(x)=f(x,x),\forall x\in X$, 2) $x_1\leq x_2\Rightarrow f(x_1,y)\leq f(x_2,y),\forall x_1,x_2,y\in X$, 3) $y_1\geq y_2\Rightarrow f(x,y_1)\leq f(x,y_2),\forall x,y_1,y_2\in X$, where $\leq$ denotes the element-wise inequality. A switched system is mixed monotone if it is mixed monotone for each mode (action) $a\in A$. 
\end{definition}
Since $\sum_{s\in S}b_t(s)=1$, (\ref{eqn:belief dynamic}) can be equivalently written as an $N-1$-dimension dynamical system
\begin{equation}\label{eqn:belief dynamic1}
b^{[1,N-1]}_{t+1} = F_a(b^{[1,N-1]}_{t})    
\end{equation}
where $b^{[1,N-1]}_{t} = [b_{t,1},...,b_{t,N-1}]^T$ and the function mapping $F_a$ will be shown in the following lemma which proves that (\ref{eqn:belief dynamic1}) is indeed mixed monotone. 
\begin{lemma}
The switched system (\ref{eqn:belief dynamic1}) is mixed monotone. 
\begin{proof}
From $(\ref{eqn:belief dynamic})$, for at $a\in A$ and $t$ we have 
\begin{equation}\label{eqn:proof1}
b_{t+1}=H_ab_t =   \begin{bmatrix}
    p_{1,1}       & \dots & p_{1,N} \\
    p_{2,1}       & \dots & p_{2,N} \\
    \vdots       & \dots & \vdots \\
    p_{N,1}       & \dots & p_{N,N}
\end{bmatrix}
\begin{bmatrix}
    b_{t,1} \\
    b_{t,2} \\
    \vdots \\
    b_{t,N}      
\end{bmatrix}
\end{equation}
Since the probabilities have to sum to one, we have $\sum_{i=1}^Np_{i,j}=1,\forall j\in\{1,...,N\}$ and $\sum_{i=1}^Nb_{t,i}=1$. Therefore, (\ref{eqn:proof1}) can be rewritten as
\begin{equation}\label{eqn:proof2}
\begin{split}
\begin{bmatrix}
    b_{t+1,1} \\
    b_{t+1,2} \\
    \vdots \\
    b_{t+1,N}      
\end{bmatrix}=
&\begin{bmatrix}
    p_{1,1}       & \dots & p_{1,N} \\
    p_{2,1}       & \dots & p_{2,N} \\
    \vdots       & \dots & \vdots \\
    1-\sum_{i=1}^{N-1}p_{i,1}& \dots & 1-\sum_{i=1}^{N-1}p_{i,N}
\end{bmatrix}\\
&\begin{bmatrix}
    b_{t,1} \\
    b_{t,2} \\
    \vdots \\
    1-\sum_{i=1}^{N-1}b_{t,i}      
\end{bmatrix}
\end{split}
\end{equation}
Since $b_{t+1,N} = 1-\sum_{i=1}^{N-1}b_{t+1,i}$, from (\ref{eqn:proof2}) we have the following equation on the $N-1$-dimensional system

\begin{equation}\label{eqn:proof3}
\begin{split}
&\begin{bmatrix}
    b_{t+1,1} \\
    b_{t+1,2} \\
    \vdots \\
    b_{t+1,N-1}      
\end{bmatrix}=\begin{bmatrix}
    p_{1,1}       & \dots & p_{1,N-1} \\
    p_{2,1}       & \dots & p_{2,N-1} \\
    \vdots       & \dots & \vdots \\
    p_{N-1,1}       & \dots & p_{N-1,N-1}
\end{bmatrix}
\begin{bmatrix}
    b_{t,1} \\
    b_{t,2} \\
    \vdots \\
    b_{t,N-1} \\     
\end{bmatrix}\\
&-\begin{bmatrix}
    p_{1,N}       & \dots & p_{1,N} \\
    p_{2,N}       & \dots & p_{2,N} \\
    \vdots       & \dots & \vdots \\
    p_{N-1,N}       & \dots & p_{N-1,N}
\end{bmatrix}
\begin{bmatrix}
    b_{t,1} \\
    b_{t,2} \\
    \vdots \\
    b_{t,N-1} \\     
\end{bmatrix}
+
\begin{bmatrix}
    p_{1,N} \\
    p_{2,N} \\
    \vdots \\
    p_{N-1,N} \\     
\end{bmatrix}
\end{split}
\end{equation}
Equivalently, we have
\begin{equation}
b^{[1,N-1]}_{t+1} = F(b^{[1,N-1]}_{t}) = A_1b^{[1,N-1]}_{t}-A_2b^{[1,N-1]}_{t}+B
\end{equation}  
where $A_k(i,j)\geq 0, \forall i,j\in \{1,...,N-1\},k\in\{1,2\}$, $B(i)\geq 0, \forall i\in \{1,...,N-1\}$. If we define $f(x,y)= A_1x-A_2y+B$, from Definition \ref{def:mixed monotone}, it is not hard to find that all the three conditions are satisfied. Since it holds for arbitrary $a\in A$, by definition, the switched system (\ref{eqn:belief dynamic1}) is mixed monotone.
\end{proof}
\end{lemma}
Mixed monotone systems  admit efficient over-approximation of the reachable set by evaluating the function $f$ at two points as proven in Theorem \ref{theorem:mixed monotone}.
\begin{theorem}\cite{coogan2015efficient}\label{theorem:mixed monotone}
Given a mixed monotone system as defined in (\ref{eqn:mixed monotone}) with decomposition function $f(x,y)$, given $x_1,x_2\in X$ with $x_1\leq x_2$, we have 
\begin{equation}\label{eqn:partition}
f(x_1,x_2)\leq F(x) \leq f(x_2,x_1), \forall x\in[x_1,x_2]    
\end{equation}
\end{theorem}
This theorem is a direct result of the mixed monotone property and is the key to the efficient abstraction, which can be seen more clearly from the following formula. 
\begin{equation}\label{eqn:efficient partition}
F([x_1,x_2])\subseteq [f(x_1,x_2),f(x_2,x_1)] 
\end{equation}
where $F(X')=\{F(x)|x\in X'\} and X'\in X$ is called the one-step reachable set from $X'$ \cite{coogan2015efficient}. $x\in [x_1,x_2]$ if and only if $x_1\leq x\leq x_2$. It can be observed from (\ref{eqn:efficient partition}) that it is sufficient to evaluate the decomposition function $f$ at two points $x_1$ and $x_2$ to compute an over-approximation of the one-step reachable set where the bounding has been shown to be tight \cite{coogan2015efficient}. 

Given the MDP model $\mathcal{M}=(S,\pi_0,A,P)$, now we are ready to construct a finite state abstraction of the  belief space dynamic as defined in (\ref{eqn:belief dynamic1}), which is similar to \cite{coogan2015efficient}. The major difference is that, the domain $X$ in \cite{coogan2015efficient} is a box  where the interval in each dimension is independent of others, while in this paper, from (\ref{eqn:belief dynamic1}) it can be seen that $X$ has the constraint $|b^{[1,N-1]}|\leq 1$.

The first step is to partition the domain $X$ into a finite set of intervals $\{I_q\},q\in Q$, where $I_q=[x^q_1,x^q_2],x_1^q\leq x_2^q$, $\bigcup_{q\in Q}I_q= X$, $int(I_q)\cap int(I_{q'})=\emptyset$, $\forall q,q'\in Q,q\neq q'$, $int(I_q)$ denotes the interior of $I_q$.

The probabilistic simplex $x_1+x_2\leq 1, x_1\geq 0, x_2\geq 0$ is gridded by squares with width 0.2. Note that the partitioned grids can have arbitrary sizes and need not to be equal. This example uses the equal size grids just for demonstration. Recall the opacity requirement (\ref{eqn:opacity}), which basically defines a bad set $\perp = \{b|\sum_{s\in S_s}b(s)>\lambda\}$ that the belief should never enter. The following lemma then shows that there exists a simple algorithm to determine whether a partition $I_q$ has an overlap with $\perp$.

\begin{lemma}
Given an interval $I_q=[x_1,x_2],x_1\leq x_2$ and the set $\perp$, then $I_q\bigcap\perp\neq\emptyset$ if  $\sum_{s\in S_s}x_2(s)>\lambda$.
\end{lemma}

Any $I_q$ that overlaps with $\perp$ is categorized as a bad region that should be avoided. Figure \ref{fig:partition} illustrates the partition of Example \ref{example} where the opacity requirement is that $b(s_1)+b(s_2)\leq 0.8$ all the time.  The blue shaded area denotes $\perp$ and all the grey shaded grids are bad regions. Therefore, we are only concerned with the $6$ non-shaded grids. We assume that the initial belief state is always outside of $\perp$. If the grid that contains the initial belief state is bad due to the overlapping, we may re-partition this grid into two smaller grids such that the initial belief state is no longer in a bad region.

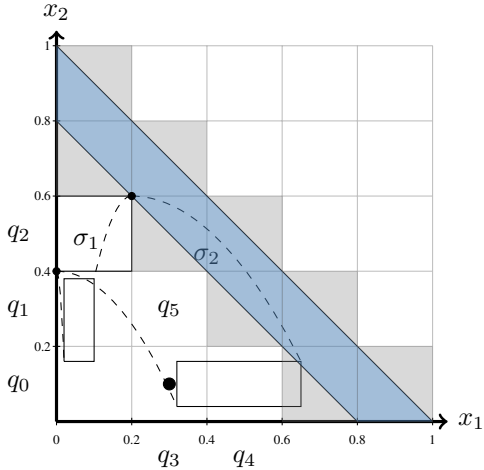
\begin{figure}[htpb]
    \centering
    \begin{tikzpicture}
    \draw[gray!50, thin, step=1] (0,0) grid (5,5);
    \draw[very thick,->] (0,0) -- (5.2,0) node[right] {$x_1$};
    \draw[very thick,->] (0,0) -- (0,5.2) node[above] {$x_2$};
    
    \draw (0,0.05) -- (0,-0.05) node[below] {\tiny 0};
    \draw (1,0.05) -- (1,-0.05) node[below] {\tiny 0.2};
    \draw (2,0.05) -- (2,-0.05) node[below] {\tiny 0.4};
    \draw (3,0.05) -- (3,-0.05) node[below] {\tiny 0.6};
    \draw (4,0.05) -- (4,-0.05) node[below] {\tiny 0.8};
    \draw (5,0.05) -- (5,-0.05) node[below] {\tiny 1};
    
    \draw (-0.05,1) -- (0.05,1) node[left] {\tiny 0.2};
    \draw (-0.05,2) -- (0.05,2) node[left] {\tiny 0.4};
    \draw (-0.05,3) -- (0.05,3) node[left] {\tiny 0.6};
    \draw (-0.05,4) -- (0.05,4) node[left] {\tiny 0.8};
    \draw (-0.05,5) -- (0.05,5) node[left] {\tiny 1};
    
    \draw (0,5) -- node[below,sloped] {} (5,0);
    
    \draw (0,4) -- node[below,sloped] {} (4,0);
    
    \draw (0.1,0.8) -- (0.5,0.8) -- (0.5,1.9) -- (0.1,1.9) -- (0.1,0.8);
    
    \draw (1.6,0.2) -- (3.25,0.2) -- (3.25,0.8) -- (1.6,0.8) -- (1.6,0.2);
    
    \draw (0,2) -- (1,2) -- (1,3) -- (0,3) -- (0,2);
    
    \draw[dashed] (0,2) parabola (0.1,0.8);
    
    \draw[dashed] (1,3) parabola (0.5,1.9);
    
    \draw[dashed] (0,2) parabola (1.6,0.2);
    
    \draw[dashed] (1,3) parabola (3.25,0.8);
    
    \node[] at  (0.4,2.4) {$\sigma_1$};
    
    \node[] at   (2,2.2) {$\sigma_2$};
    
    \node[] at   (-0.5,0.5) {$q_0$};
    \node[] at   (-0.5,1.5) {$q_1$};
    \node[] at   (-0.5,2.5) {$q_2$};
    
    \node[] at   (1.5,-0.5) {$q_3$};
    \node[] at   (2.5,-0.5) {$q_4$};
    \node[] at   (1.5,1.5) {$q_5$};
    
    \fill[blue!50!cyan,opacity=0.3] (0,5) -- (5,0) -- (4,0) -- (0,4) -- cycle;
    
    \fill[red!50!cyan,opacity=0.3] (0,5) -- (1,5) -- (1,3) -- (0,3) -- cycle;
    
    \fill[red!50!cyan,opacity=0.3] (1,4) -- (2,4) -- (2,2) -- (1,2) -- cycle;
    
    \fill[red!50!cyan,opacity=0.3] (2,3) -- (3,3) -- (3,1) -- (2,1) -- cycle;
    
    \fill[red!50!cyan,opacity=0.3] (3,2) -- (4,2) -- (4,0) -- (3,0) -- cycle;
    
    \fill[red!50!cyan,opacity=0.3] (4,1) -- (5,1) -- (5,0) -- (4,0) -- cycle;
    \draw[black,fill=black] (1.5,0.5)  circle (.5ex);
    
    \draw[black,fill=black] (0,2)  circle (.3ex);
    \draw[black,fill=black] (1,3)  circle (.3ex);

    %\foreach \x in {0,1,2,3,4,5} \draw (\x,0.05) -- (\x,-0.05) node[below] {\tiny(\x*0.2};
    %\foreach \y in {0,0.2,0.4,0.6,0.8,1} \draw (-0.05,\y) -- (0.05,\y) node[left] {\tiny\y};

    %\fill[blue!50!cyan,opacity=0.3] (8/3,1/3) -- (1,2) -- (13/3,11/3) -- cycle;

    %\draw (-1,4) -- node[below,sloped] {\tiny$x_1+x_2\geq3$} (5,-2);
    %\draw (1,-3) -- (3,1) -- node[below left,sloped] {\tiny$2x_1-x_2\leq5$} (4.5,4);
    %\draw (-1,1) -- node[above,sloped] {\tiny$-x_1+2x_2\leq3$} (5,4);
   % \draw[step=0.1,black,thin] (0,0) grid (1,1);
    \end{tikzpicture}
    \caption{A gridded partition, the opacity requirement is $x_1+x_2\leq 0.8$. The big dot denotes the initial condition.}
    \label{fig:partition}
    %\vspace{-3.5mm}
\end{figure}

The second step is to construct the NFA $\mathcal{T}=(Q\bigcup bad,\Sigma,\delta,I)$ given the MDP model $\mathcal{M}=(S,\pi_0,A,P)$ and the partition $\{I_q\},q\in Q$, where $\Sigma=A$. To determine the transition relation in $\mathcal{T}$, $q'\in\delta(q,\sigma)$, if and only if $[f_\sigma(x_1,x_2),f_\sigma(x_2,x_1)]\bigcap I_{q'}\neq \emptyset$. That is,if our over-approximated one-step reachable set for $I_q$ has a non-empty interception with the partitioned region $I_q'$ given the action $\sigma$, there will be a transition relation $q'\in\delta(q,\sigma)$ in the abstraction system $\mathcal{T}$. We still take Figure \ref{fig:partition} as the example to illustrates how to determine the transition relation. All the shaded grids are bad regions and there are $6$ states (correspondingly $6$ regions) of interest in $\mathcal{T}$. Let's look at $q_2$. By mixed monotone property, we only have to evaluate two points, namely $p_1=(0,0.4)$ and $p_2=(0.2,0.6)$. With action $\sigma_1\in\Sigma$, from Figure \ref{fig:partition} it can be seen that the over-approximation reachable set overlaps with $q_0$ and $q_1$. Therefore, we have $q_0\in\delta(q_2,\sigma_1)$ and $q_1\in\delta(q_2,\sigma_1)$. %$(q_2,\sigma_1,q_0)\in\rightarrow$ and $(q_2,\sigma_1,q_1)\in\rightarrow$. 
Similarly, we have 
$q_3\in\delta(q_2,\sigma_2)$, $q_4\in\delta(q_2,\sigma_2)$ and $bad\in\delta(q_2,\sigma_2)$. 
%$(q_2,\sigma_2,q_3)\in\rightarrow$,$(q_2,\sigma_2,q_4)\in\rightarrow$  and $(q_3,\sigma_2,\text{bad})\in\rightarrow$.
Here, $bad$ denotes a bad region.

\begin{figure}
		\centering	
\begin{tikzpicture}[shorten >=1pt,node distance=2cm,on grid,auto, bend angle=20, thick,scale=1, every node/.style={transform shape}] 
				\node[state] (s0)   {$q_0$}; 
				\node[state] (s1) [above = of s0] {$q_1$}; 
				\node[state] (s2) [left= of s1]  {$q_2$}; 
				\node[state,initial] (s3) [below = of s2] {$q_3$}; 
				\node[state] (s4) [right= of s1] {$q_4$}; 
				\node[state] (s5) [below= of s4] {$q_5$}; 
		
				\path[->]
				(s0) edge [bend right=15]  node {} (s1)
				(s0) edge [loop below]  node {} ()
				(s1) edge [bend right=15]  node {} (s0)
				(s1) edge [loop above]  node {} ()

				(s2) edge   node {} (s1)
				(s2) edge   node {} (s0)
				(s3) edge [bend left=10]  node {} (s1)
				(s3) edge   node {} (s0)
				(s4) edge [bend left=10]  node {} (s1)
				(s4) edge [bend right=45]  node {} (s2)
				(s5) edge   node {} (s1)
				(s5) edge   node {} (s0)
			
				; %end path 	
				\path[->,dashed]

				(s1) edge [bend left=10]  node {} (s3)
				(s1) edge [bend left=10]  node {} (s4)
				
				(s5) edge [bend left=45]  node {} (s3)
				(s5) edge  node {} (s4)				
				; %end path 
				\end{tikzpicture} 
		\caption{$\mathcal{T}$ after pruning, the solid lines denote the transitions induced by $\sigma_1$ and dashed lines for $\sigma_2$ }\label{fig:abstraction}
		   \vspace{-4.5mm}
	\end{figure}
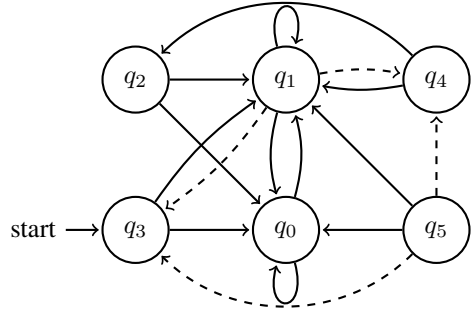

It should be noted that such abstraction could produce spurious trajectories that do not actually exist in (\ref{eqn:belief dynamic1}). This is generally unavoidable in the partition based approaches. However, since we are only interested in the safety property in the belief space (if bad belief state is reachable), such spuriousness may make the results more conservative, but does not affect its correctness, as all the transitions that are possible to happen in the concrete system (\ref{eqn:belief dynamic1}) are included in the abstraction system. 

Any outgoing transition $(q,\sigma,bad)$ should be deleted from $q$. To do this, we directly disable the action $\sigma$ from $q$, as the transitions are nondeterministic. For example, in Figure \ref{fig:partition}, since we have $(q_3,\sigma_2,bad)\in\rightarrow$, action $\sigma_2$ will be disabled in $q_3$. If such pruning results in any  state $q'$ blocking, that is, all its outing transitions for all actions are pruned, then $q'$ and all its incoming and outgoing transitions are deleted.  Such process continues until no more states are pruned from $\mathcal{T}$ or the initial state of $\mathcal{T}$ is pruned. If the latter situation happens, it implies that the current partition may be too coarse so that the over-approximation is too conservative, which we may need to find a finer partition scheme, for example, by having smaller grids. It could also be the case that the belief dynamics (\ref{eqn:belief dynamic1}) will eventually drive the belief state to $\perp$ under arbitrary switching. If this is the case, there is no hope to find a non-empty $\mathcal{T}$ after pruning, regardless of how the belief space is partitioned. Determining whether it is true relies on the reachability analysis of the underlying switched linear systems and is out of the scope of this paper.  The resulting NFA from the griding in Figure \ref{fig:partition} is shown in Figure \ref{fig:abstraction}. %For clarity, the solid lines denote the transitions induced by $\sigma_1$ and dashed lines for $\sigma_2$.

\section{Controller Synthesis}\label{sec:synthesis}
Once we obtain the abstracted belief model $\mathcal{T}$, together with the MDP model $\mathcal{M}$, it is then possible to synthesize a policy that simultaneously satisfies the task and privacy specification, regardless of how the nondeterminisim in the abstracted belief model is resolved. We propose two different solutions based on different capabilities of the intruders.

\subsection{Direct Synthesis}
For Example \ref{example}, if the intruder is the supplier which can observe the purchasing actions since the purchase has to go through it, we need a purchasing strategy such that the supplier may never be sure with high confidence that the company's inventory is running too low or too high. We take two steps to obtain synthesize the policy. The first step is to obtain a new MDP $\mathcal{M}'$ based on the original model $\mathcal{M}$ to constrain the available actions at each state considering the opacity constraint. Recall that we assume that in $\mathcal{M}$, $A(s)=A,\forall s\in S$. However, with the privacy constraints represented as $\mathcal{T}$, some of the actions may cause privacy violation (even though not necessarily, since $\mathcal{T}$ is an over-approximation of the concrete dynamics). The new MDP is  $\mathcal{M}'=(S',\pi_0',A',P')$ where $S'=S,\pi_0'=\pi_0,A'=A,P'=P$, the only difference is $A'(s)\subseteq A(s), \forall s\in S$. To obtain $A'(s)$, we propose to product the NFA $\mathcal{T}_\mathcal{M}$ obtained from $\mathcal{M}$ and $\mathcal{T}$. The synchronous product is defined in a standard way as follows \cite{baier2008principles}.
 \begin{definition}[Synchronous Product of NFAs]\label{def:parallel}
	Given two NFAs $ \mathcal{T}_i=(Q_i,\Sigma,\delta_i,I_i)$ with $i=1,2$, the product automaton $\mathcal{T}$ as the result of synchronous product of $\mathcal{T}_1$ and $\mathcal{T}_2$ is the NFA $\mathcal{T}=\mathcal{T}_1\otimes\mathcal{T}_2= (Q,\Sigma,\delta,I)$, where $Q= Q_1\times Q_2,I = I_1\times I_2$ and $(q_1',q_2')\in\delta((q_1,q_2),\sigma)$ if and only if $q_1'\in\delta_1(q_1,\sigma)$ and $q_2'\in\delta_2(q_2,\sigma)$.	
	%$((q_1,q_2),\sigma,(q_1',q_2'))\in\rightarrow$ if and only if $(q_1,\sigma,q_1')\in\rightarrow_1$ and $(q_2,\sigma,q_2')\in\rightarrow_2$.
\end{definition}
Once we get $\mathcal{T}'=\mathcal{T}_\mathcal{M}\otimes\mathcal{T}=\{Q',\Sigma,\delta',I'\}$, we obtain $A'(s)$ as follows.
\begin{equation}\label{eqn:get A'(s)}
    A'(s) = \bigcap_{(q,s)\in Q'} \Sigma((q,s))
\end{equation}
where $\Sigma((q,s)\subseteq\Sigma$ denotes the set of actions available at the state $(q,s)$. Intuitively, $A'(s)$ denotes all the actions at $s$ that are guaranteed to preserve privacy at any time. With (\ref{eqn:get A'(s)}), we obtain $A'(s)=\{\sigma_1\},\forall s\in S$ in Figure \ref{fig:abstraction}. Then the second step is the controller synthesis performed on the MDP $\mathcal{M}'$ to obtain the policy such that the task specification $\phi$ in LTL or PCTL can be satisfied with the optimal probability $p$ on $\mathcal{M}'$. For this step, the synthesis algorithm can be found in \cite{baier2008principles}.
\begin{theorem}
The optimal policy obtained on $\mathcal{M}'$ satisfies the opacity specifications and incurs the same probability to satisfy the  specification as in $\mathcal{M}$.
\end{theorem}
\begin{proof}
Since $\mathcal{M}'$ only differs from $\mathcal{M}$ in the available actions at each state, it is straightforward to see that the same policy induces the same probability on both of the MDP models. As for opacity specification, from (\ref{eqn:get A'(s)}), we are guaranteed to stay in the ``safe'' belief space since the action being enabled belongs to $A'(s)$ and any action selected from $A'(s)$ is safe regardless of the current abstracted belief partition the system is in. 
\end{proof}
Note that it could be the case that some states in $\mathcal{M}'$ do not have any action available, in such a case, an iterative pruning process is applied to delete such blocking states until either there is no more state to prune or one of the initial state is pruned. If it is the latter case, we may need a finer partition to make the abstraction less conservative.  

\subsection{Edit function}
If the intruder is the competitor in Example \ref{example},  it is then possible to manipulate the purchase activity report observable to it, such that the competitor may never  infer with high confidence of the company's inventory level being too low or too high. Unlike suppliers, the competitor cannot distinguish between the real or the reported purchase. This approach is inspired by the edit function synthesis in \cite{wuy2017synthesis} where the system has the capability to modify the observations of the intruder based on the real system action, such that the observed behavior is consistent with the model's behavior and at the same time, the intruder may never determine with certainty that the current state is a secret state. Formally, given an MDP $\mathcal{M}=(S,\pi_0,A,P)$ and its corresponding NFA $\mathcal{T}_\mathcal{M}=(Q,\Sigma,\delta,I)$ where $Q=S,\Sigma=A$, we are looking for an edit function $f_e:\Sigma^*\rightarrow\Sigma^*$, such that the followings are satisfied.
\begin{enumerate}
    \item $\forall\omega\in\mathcal{L}(\mathcal{T}_{\mathcal{M}}),f_e(\omega) \text{ is defined}$
    \item $\forall\omega\in\mathcal{L}(\mathcal{T}_{\mathcal{M}}),\exists s\in I, \delta(s,f_e(\omega)) \text{ is defined}$
    \item $\forall\omega\in\mathcal{L}(\mathcal{T}_{\mathcal{M}})$, after executing $f_e(\omega)$, the switched system defined as in (\ref{eqn:belief dynamic}) satisfies (\ref{eqn:opacity}).
\end{enumerate}
Intuitively, the first item requires that the edit function $f_e$ should be defined for all the possible behaviors of the system. The second item requires that the output of the edit function, which is observed by the intruder, should also be a valid behavior of the system. The third item requires that the output behavior of the edit function should satisfy the opacity specification. Note that from this definition, $f_e$ may not be unique.

%$\mathcal{T}=(Q\bigcup bad,\Sigma,\delta,I)$

Given a system modeled as an MDP $\mathcal{M}=(S,\pi_0,A,P)$, $f_e$ can be implemented as a (potentially) infinite-state edit automaton $\mathcal{T}_{f}=(Q_f,\Sigma,\delta_f,I_f)$, where $\Sigma=A$, $\delta_f\subseteq Q_f\times\Sigma\times\Sigma^*\times Q_f$. Therefore, each transition $(q,\sigma,o,q')$ in $\mathcal{T}_{f}$ denotes that from state $q$, when $\sigma\in\Sigma$ actually happens in the system, it is modified to become $o\in\Sigma^*$ which is observed by the intruder, and then the edit automaton transits to some $q'$. Intuitively, if we edit every possible executions to be the empty string $\epsilon$, the intruder will observe nothing and the system will always be opaque if it is opaque initially. However, such case may become trivial. Therefore, we restrict the transitions of the edit automaton to be of the form $\delta_f\subseteq Q_f\times\Sigma\times\Sigma\times Q_f$, that is, it must output one and only one event $\sigma\in\Sigma$ after an event has actually happened in the system.

In this paper,  $f_e$ is easier to synthesize, since all actions are defined at every state, and the second requirement of $f_e$ is automatically satisfied.  To guarantee the third requirement, the output behavior of $f_e$ can be the language generated by the abstraction $\mathcal{T}=(Q,\Sigma,\delta,I)$. That is, the edit automaton $\mathcal{T}_{f}=(Q,\Sigma,\delta_f,I)$, where given the transition $(q,o,q'),o\in\Sigma$ in $\mathcal{T}$ and given the actual event $\sigma\in\Sigma$, there is a transition $(q,\sigma,o,q')$. Therefore, the intruder observes a subset of the generated language of $\mathcal{T}$, which is guaranteed to preserve opacity. Furthermore, since we don't have any restriction on the actual event $\sigma$, the requirement $1$ of the edit function is also satisfied.

In our example, the observation function is essentially the abstracted model $\mathcal{T}$ in Figure \ref{fig:abstraction}. Regardless of the real system action  that is executed, starting from $q_0$, the edit function may select any action $\sigma$ that is defined at the current belief region $q$ to be the observation to the intruder. Then the next abstracted belief state $q'$ is determined by the belief dynamic (\ref{eqn:belief dynamic1}). Note that such update is based on the ``fake'' action $\sigma$, not the real system action, which is hidden by the edit function. For example, starting from $q_0$ in Figure \ref{fig:abstraction}, the edit function may output $\sigma_1$, regardless of actually event $\sigma_1$ or $\sigma_2$ happened. If the belief state update based on the output behavior $\sigma_1$ results in $q_1$, from Figure \ref{fig:abstraction}, next time it could either output $\sigma_1$ or $\sigma_2$ irrespective of actual event.

In this approach, it can be observed that the privacy enforcement and the controller synthesis are decoupled. We could separately obtain the edit function from $\mathcal{T}$ and synthesize the optimal policy for a given specification. Therefore, the advantage of this approach comparing to the direct synthesis is that the optimal performance can always be obtained, regardless of the privacy constraint, while in direct synthesis, the available actions at each state are limited by (\ref{eqn:get A'(s)}).

\section{Conclusion}\label{sec:conclusion}
In this paper we proposed, to the best our knowledge, a new notion of opacity defined on the belief space. We then proposed two approaches to synthesis privacy preserving controllers that regulates the MDP model, so that  the privacy can be preserved. Both approaches rely on the abstracted model on the belief space, where we proved that the belief dynamic is mixed monotone and thus efficient abstraction algorithm exists. %Since both the belief abstraction and direct synthesis may bring conservativeness to the results, o
Our future work will be focusing on exploring less conservative approaches to guarantee the privacy and task accomplishment.

\bibliographystyle{IEEEtran}
\bibliography{root}
\end{document}